\documentclass{article}
\usepackage{amsmath}
\usepackage{amssymb}
\usepackage[ruled,lined]{algorithm2e}
\usepackage{enumerate}
\usepackage{graphicx}
\usepackage{color}
\newtheorem{definition}{Definition:}
\newtheorem{lemma}{Lemma:}
\newenvironment{proof}{{\bf Proof:}}{\hfill {\fbox{}}}

\begin{document}
\title{BITS-Tree -- An Efficient Data Structure for Segment Storage and Query Processing}
\author{K.S. Easwarakumar\thanks{Corresponding Author.Email: easwara@cs.annauniv.edu} and T.Hema\thanks{Email: hema@cs.annauniv.edu}\\Department of Computer Science \& Engineering\\ Anna University, Chennai 600 025, INDIA.}
\maketitle
\begin{abstract}
In this paper, a new and novel data structure is proposed to dynamically insert and delete segments. Unlike the standard segment trees\cite{b1}, the proposed data structure permits insertion of a segment with interval range beyond the interval range of the existing tree, which is the interval between minimum and maximum values of the end points of all the segments. Moreover, the number of nodes in the proposed tree is lesser as compared to the dynamic version of the standard segment trees, and is able to answer both stabbing and range queries practically much faster compared to the standard segment trees. 
\end{abstract}  
\section{Introduction}
Solving geometrical problems computationally seems to be difficult due to its complex nature such as varying dimensionality and shape. Several algorithms are available in literature\cite{bcko,ps} for storing and retrieving geometrical objects. Segment tree \cite{b1} is one such data structure designed to handle intervals on the real line, whose extremes belong to a fixed set of abscissas. The static segment tree requires $O(n\log n)$ space, and the stabbing query processing time is $O(k+\log n)$, where $n$ is the number of segments and $k$ is the number of output segments. Though it looks reasonable in space and query processing, it does not support dynamic insertion and deletion of segments. However, in the dynamic segment trees, one can insert or delete a segment in the existing segment tree, but storage space requirements are likely get increased as the range of the tree is not based on the segment end points. The range of the tree is predefined, and one can insert or delete the segments whose interval contained in this range, thus in those cases where the range of the segment falls beyond the range of the segment tree, reconstruction of the tree is the only possibility, which is not advisable. Also, this may practically increases the time as the number of nodes to be traversed get increased.  
\par
Variants of segment trees are used in packet classification problem. One such data structure is Fat Inverted Segment Trees (FIS-trees) \cite{fm,t} with the space requirements of $O(n^{1+\frac{1}{l}})$ and having complexity for insertion/deletion as $O(n^\frac{1}{l}\log n)$, where $l$ is the height of the tree. Here, the tree is compressed  (made `fat') by increasing the degree to more than two in order to decrease the height, however this has an upper bound on the total number of insertions and deletions allowed. Also, Agarwal et. al \cite{aa} proposed a linear-size data structure for the stabbing semi-group problem by combining the features of interval and segment tree, with the assumption that the end points of all intervals belong to a fixed set of points. 
\par
In this paper, a novel data structure, BITS-tree (\textbf{B}alanced \textbf{I}norder \textbf{T}hreaded \textbf{S}egment Tree),is proposed to dynamically insert and delete segments, with interval range of the node is only based on the end points of segments stored in the tree. In addition to this, it answers both stabbing and range queries efficiently. Unlike in the dynamic version of the standard segment trees, the root of the BITS-tree does not have the information of overall interval range of nodes in the tree, the BITS-tree also permits insertion of segment with any interval range, and thus it is a real dynamic tree that is suitable for all situations. 
\section{BITS-Tree}
\begin{definition}
A $BITS$-tree is a height balanced two-way inorder-threaded binary tree $T$ that satisfies the following properties.
\begin{enumerate}
\itemsep=0pt
\parskip=0pt
	\item Each node $v$ of $T$ associates with a range $r$ and a list of segments containing the range $r$.
	\item The range of a node cannot overlap with a range of any other node, other than at the end points. 
	\item The range of the nodes are sorted according to the inorder sequence. 
	\item It has a special node, called $dummy$, with range $\phi$ (empty), and its list is $\phi$. 
	\item The inorder predecessor of the first node of the inorder sequence is $dummy$. Similarly the inorder successor of the last node of the inorder sequence is $dummy$.
\end{enumerate}
\end{definition}
\par
For example, consider the $BITS$-tree given figure \ref{f1}.
\begin{figure}[!ht]
\begin{center}
\input{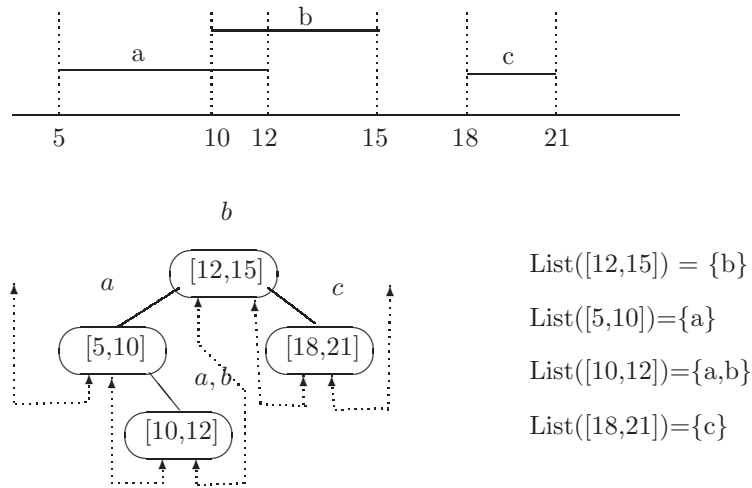}
\end{center}
\caption{A $BITS$-tree.}
\label{f1}
\end{figure}
The inorder sequence of range is  $[5,10]$, $[10,12]$, $[12,15]$ and $[18,21]$, which is sorted in ascending order, as this satisfies the binary search tree properties on ranges. So, searching for a particular range could be done as in the case of binary search tree, using the following definition. The doted lines in figure \ref{f1} denote threads. The hanging threads points to the $dummy$ node.
\begin{definition}
\label{d1}
Given two segments $S=[p,q]$ and $T=[m,n]$,  $S < T$ implies $q\le m$, that is the segment $S$ falls left of $T$. Similarly,  $S > T$ implies $n\le p$, and this means $S$ falls right of $T$.
\end{definition}
\begin{definition}
\label{d2}
Given two segments $S=[p,q]$ and $T=[m,n]$, where $p\neq q$ and $m\neq n$,  $p<m<q<n$ denotes $S$ {\em left overlaps} with $T$. Similarly, $m<p<n<q$ denotes $S$ {\em right overlaps} with $T$. Otherwise, $S$ is said to be {\em contained in} $T$ when $m<p<q\le n$ or $m\le p<q<n$, or {\em covered up} $T$ when $p<m<n\le q$ or $p\le m <n<q$.
\end{definition}
\par All four cases stated in definition \ref{d2} are commonly called as $S$ and $T$ {\em overlaps}. Note that, $S$ left overlaps with $T$ also implies $T$ right overlaps with $S$. Similarly, $S$ contained in $T$ also implies $T$ covered up $T$.
\begin{definition}
Given two overlapping segments $S=[p,q]$ and $T=[m,n]$, $S\cup T$ and $S \cap T$ are respectively defined as $[\min\{p,m\},\max \{q,n\}]$ and  $[\max\{p,m\},\min \{q,n\}]$.
\end{definition}
\begin{definition}
\label{d3}
Given two overlapping segments $S$ and $T$, $S\cup T$, say $[a,b]$, partitions into $L(S,T)$, $C(S,T)$ and $R(S,T)$. Here, $C(S,T) = S \cap T$, say $[p,q]$. Now, 
\begin{center}
$L(S,T)= \left \{ \begin{array}{lcl} [a,p]&~~& if~ a\neq p\\ \phi &&  Otherwise\end{array}\right .$\\
and ~~
$R(S,T)= \left \{ \begin{array}{lcl} [q,b]&~~& if~ q\neq b\\ \phi &&  Otherwise\end{array}\right .$
\end{center}
\end{definition}
\begin{definition}
\label{d4}
Let $S=[p,q]$ and $T=[m,n]$ be the two overlapping segments and their associated lists respectively be $list(S)$ and $list(T)$. Now, the list associated with $C(S,T)$, $L(S,T)$ and $R(S,T)$ are defined as follows.
\begin{eqnarray*}
list(C(S,T)) & = & list(S) \cup list (T)\\
list(L(S,T)) &= &\left \{ \begin{array}{lcl} list(S)&~~&if~p<m\\
                                             list(T)&&if~p>m\\
                                             \phi&&Otherwise
                           \end{array} \right .\\
list(R(S,T))&=&\left \{ \begin{array}{lcl} list(S)&~~&if~n<q\\
                                             list(T)&&if~n>q\\
                                             \phi&&Otherwise
                           \end{array} \right .
\end{eqnarray*}
\end{definition}
\subsection{Insertion}
For to insert a segment, say $S$, one has to search for a node, whose range overlaps with $S$, from root like in the binary search tree based on the relation given in definition \ref{d1}. As soon as, an overlapping range is found in any node $v$, the insertion process starts as follows. Let $R$ be the range of $v$. Now, $S \cup R$ is partitioned into sub-ranges as described in definition \ref{d3}. Now, range of $v$ will be changed as $C(S,R)$ and the segment $S$ will be included in the list of $v$. Further, insertion of $L(S,R)$ (if not empty) continues with the inorder predecessor of $v$, and insertion of $R(S,R)$ (if not empty) continues with the inorder successor of $v$. However, in these two cases, the list of segments to be inserted may change, as in definition \ref{d4}. Note that initially, the list of segments to inserted contains only $S$, and later may be changed according to definition \ref{d4}. At one point of time, the insertion may reach a null pointer (not always). If so, the balancing range (sub-range of segment not yet considered) will be considered as a range for a node, by creating a new one, with necessary thread pointers and having the list currently to be inserted. As the tree is balanced, creation of a new node may require necessary rotation to get the tree to be balanced. The rotations we performed here is only the AVL-tree based rotations \cite{avl}. The formal description of the insertion procedure is given in algorithm \ref{a1}.
\begin{algorithm}
\caption{BITS-Insert($T$,$S$,$L$)}\label{a1}
\KwIn{$T$ - the pointer to the root of the tree, $S$ - the segment to be inserted; and $L$ - Set of segments, initially $L=\{S\}$.}
\hspace{-3mm}$P\leftarrow \phi$;\\
$[p,q]\leftarrow S$;\\
\If{$T\neq NULL$ and $T\neq HEAD$}
{
  \Repeat{$C\neq \phi$ {\rm or} $T=NULL$}
  {
dd     \hspace{-3mm}$[m,n]\leftarrow RANGE(T)$;\\
     $C\leftarrow [m,n]\cap [p,q]$;\\
     $P\leftarrow T$;\\
     \If{$C = \phi$}
     {
        \hspace{-3mm}\lIf{$[p,q] <[m,n]$}{$T \leftarrow LCHILD(T)$;}\\
        \lElse{$T\leftarrow RCHILD(T)$;}
     } 
  }
  \If {$C \neq \phi$}
  {
    \hspace{-3mm}$RANGE(T) \leftarrow C$;\\
    $R \leftarrow LIST(T)$;\\
    $LIST(T) \leftarrow LIST(T)\cup \{[p,q]\}$;\\
    \If {$p\neq m$}
    {
       \hspace{-3mm}\lIf{$p<m$}{$R\leftarrow \{[p,q]\}$;}\\
       BITS-Insert($PRED(T),[\min \{p,m\},\max \{p,m\}],R$);
    }
    \If {$q\neq n$}
    {
       \hspace{-3mm}\lIf{$q>n$}{$R\leftarrow \{[p,q]\}$;}\\
       BITS-Insert($SUCC(T),[\min \{q,n\},\max \{q,n\}],R$);
    }
  }
}
\If{$T = NULL$ or $T=HEAD$}
{
   \hspace{-3mm}$N\leftarrow CreateNode()$;\\
   $RANGE(N)\leftarrow [p,q]$; \\
   $LIST(N) \leftarrow L$;\\
   \If{$P\neq \phi$}
   {
     \hspace{-3mm}\lIf{$RANGE(P) < RANGE(N)$}{$RCHILD(P)\leftarrow N$;}\\
     \lElse{$LCHILD(P) \leftarrow N$;}
   }
   Perform rotation if required;
}
\end{algorithm}
The CreateNode() function used in algorithm \ref{a1}, create a node as in threaded binary tree \cite{k,pt}, however here the treads are bidirectional.
\begin{figure}[!ht]
\begin{center}
\input{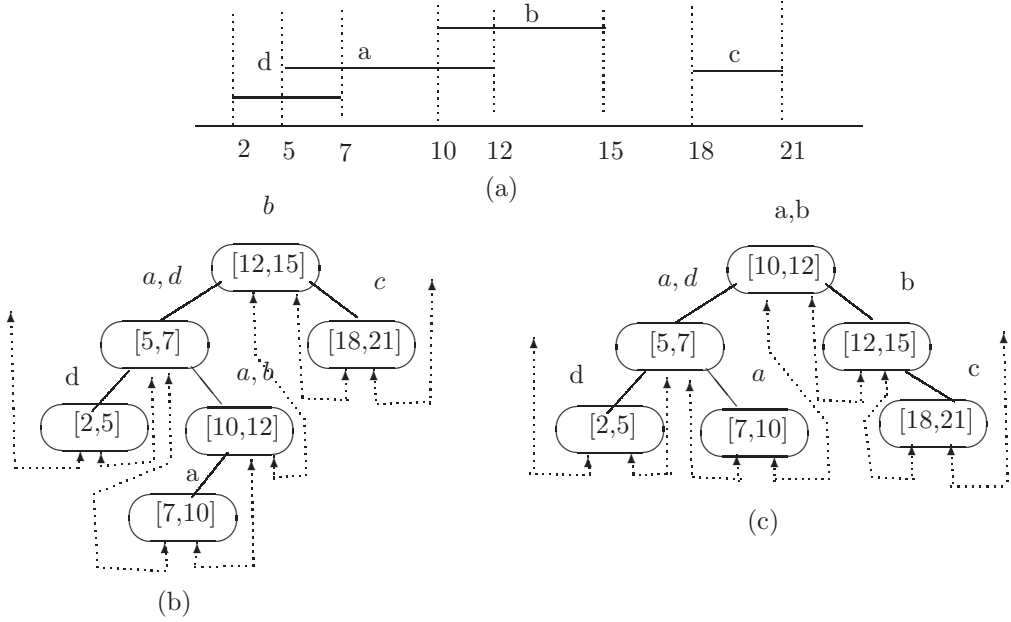}
\end{center}
\caption{Insertion: (a) Set of segments including new segment $d$. (b) Insertion of segment $d$ in BITS-tree in figure \ref{f1}, before rotation. (c) BITS-tree after rotation.}
\label{f2}
\end{figure}
\par
For example, consider insertion of segment $d=[2,7]$ in the BITS-tree given in figure \ref{f1}. The search starts with the root node, and stop in the node with range $r=[5,10]$ as this range overlaps with $d$, and $d \cup r$ results in $[2,5]\cup [5,7]\cup [7,10]$, as in definition \ref{d3}. Here, $L(d,r) = [2,5]$, $C(d,r)= [5,7]$ and $R(d,r)= [7,10]$. Now, the range $[5,7]$ is  retained at that node by including $d$ in the list of this node, and further insertion takes with respect to the ranges $[2,5]$ and $[7,10]$ at the inorder predecessor and successor, respectively. This process in turn creates two new nodes in the tree given in figure \ref{f1}, one with range $[2,5]$ and other with range $[7,10]$. This is shown in figure \ref{f2}(b). As creation of new node with range  $[2,5]$ does not affect the balancing factor, no rotation is required. However, creation of new node for the $[7,10]$ affects the balancing factor (figure \ref{f2}(b)), and hence necessary rotation to be carried out, here it is left-right rotation. The effect of this rotation is shown in figure \ref{f2}(c). Note here that the thread pointers are not affected due to rotation. Whichever be the rotation, the thread pointer cannot change, as the rotation does not affect the inorder sequence. Note here that the usage of threads obviously reduces the time for inserting a segment, as compared to the standard dynamic version of the segment tree. \par
The time required for inserting a node is $O(\log n+k)$, where the first factor $\log n$ stands for locating first overlapping node, and $k$ be the number of nodes in which the new segment get inserted.  Note that, a segment tree contains at most $2n-1$ nodes, where $n$ is the number of segments in the tree. 
\subsection{Deletion}  
Let $S=[p,q]$ be the segment to be deleted. First step of the deletion process is to locate a node $v$ with range $[p,-_q]$, where $-_q$ denotes a value less than or equal to $q$. If the list of the node $v$ does not contain the segment $S$, then deletion does not require, and the process of deletion can be terminated. Otherwise, the segment $S$ can be removed from the list of $v$. Now, after removal of $S$, if $LIST(v)$ is empty then convert the $RANGE(v)$ as $\phi$ (empty), and $v$ should be merged with $PRED(v)$. Note that, during merge the node appears at lower level will be joined at the node at the higher level, and also it performs necessary rotation if required to make the tree balanced. In case, the $LIST(v)$ is not empty after removal of $S$ then compare the $LIST(v)$ with the $LIST(PRED(v))$ and if found matching then merge the node $v$ with $PRED(v)$. Now, the deletion process can be continued with the $SUCC(v)$, and that to be continued till the node $w$ with range $[_p-,q]$, where $_p-$ denotes a value greater than equal to $p$. Then, to complete the deletion, compare the new range of $w$ with the $SUCC(w)$, and if found matching then $w$ and $SUCC(w)$ must be merged. Each time when a node is merged with the inorder predecessor or successor, necessary rotation need to be performed as in the case of AVL trees.
\par
\begin{algorithm}
\caption{BITS-Deletion}\label{a2}
\KwIn{T - the pointer to the root of the tree, and $[p,q]$ - the segment to be deleted.}
\Begin{
\While {$RANGE(T) \neq [p,-_q]$}
{
    \hspace{-3mm}\lIf {$RANGE(T) < [p,q]$}
      {$T \leftarrow RCHILD(T)$;\\}
    \lElse
      {$T \leftarrow LCHILD(T)$;}
}
\Repeat{$RANGE(T)=[_p-,q]$}
{\hspace{-3mm}$LIST(T) \leftarrow LIST(T)-\{[p,q]\}$;\\
\lIf{$LIST(T)=\phi$}{$RANGE\leftarrow \phi$;\\}
\If{$LIST(T)=LIST(PRED(T))$ {\rm or} $LIST(T)=\phi$}{merge(T,PRED(T));}
\lElse{$T \leftarrow SUCC(T)$;}
}
\lIf{$LIST(T)=LIST(SUCC(T))$}{merge(T,SUCC(T));}
}
\end{algorithm}
\begin{figure}[!ht]
\begin{center}
\input{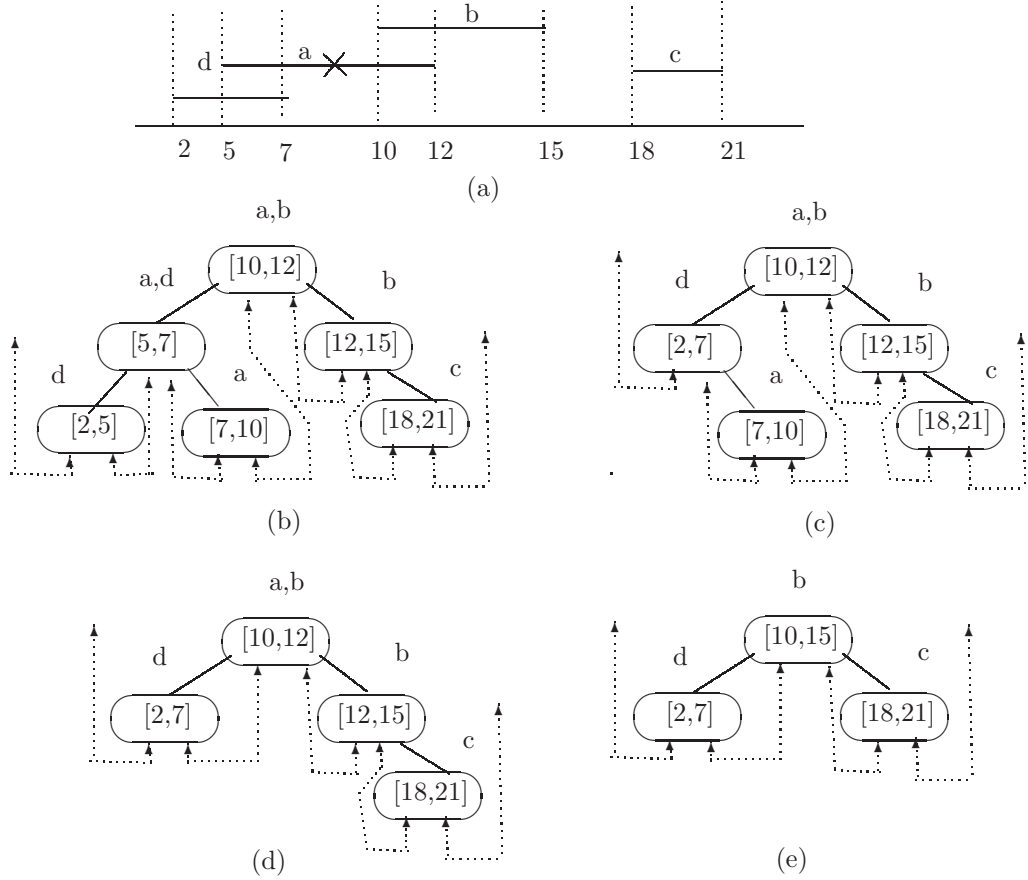}
\end{center}
\caption{Deletion: (a) Set of segments including the segment $a$ to be deleted. (b) Original $BITS$-tree. (c) After removal of segment $a$ from node having range [5,7]. (d) After removal of segment $a$ from node having range [7,10]. (e) After removal of segment $a$ from node having range [10,12].}
\label{f3}
\end{figure}
As an example, suppose segment $a=[5,12]$ is removed, as in figure \ref{f3}(a).  This process starts with node having range $[5,7]$. Now, after removal of segment $a$ from this node, the list of this node matches with its inorder predecessor, thus to be merged with its inorder predecessor and the result is shown in figure \ref{f3}(c). Next, the process continues with the node having range $[7,10]$. Now by removing segment $a$ from this node, its list becomes empty and so its range to be modified as empty($\phi$) and required merging is to be done with its inorder predecessor, and the result is shown in figure \ref{f3}(d). Finally, the segment $a$ must be removed from the list of the node having range $[10,12]$. Here, the list of this node does not match with its inorder predecessor, and hence merging with the predecessor is not required, rather as a last step of the deletion process, its range is to be compared with the inorder successor, and as it is matching it should be merged with its in order successor. The final result is shown in figure \ref{f3}(e).\
\par
Like insertion, deletion also takes $O(\log n + k)$ time, where $k$ is the number of nodes in which the segment to be deleted is present. 
\subsection{Query Processing} 
\subsubsection{Stabbing Query}
Given a point (one dimensional) $p$, it reports the segments containing $p$. On the other hand, it reports set of all segments intersects with the line $x=p$. In standard segment \footnote{trees}, the answer to this query lies in several nodes and one has to take union of lists of all those nodes to report the final answer. However, in BITS-tree, the answer lies in only one node. Here, the process is only to find the node whose range contains $p$, and the list of that node is the answer to the query. As the ranges of the nodes are not overlapped, a simple binary search will do for locating the range contains $p$. Since, BITS-tree is height balanced, $O(\log n)$ time is sufficient for locating a node with such range. Due to reduction in number of nodes and height as compared to the standard segment tree, our approach certainly outperforms considerably as compared to the existing one.
\subsubsection{Range Query}
Knowing set of segments overlaps with a given segment is an important problem in computational geometry\cite{bcko}. In BITS-tree, such queries can be answered very efficiently as compared to the standard segment trees. Here, given a segment $s=[p,q]$, finding set of overlapping segments are determined by first locating a node $u$ having range perfectly contains $p$, and traverse through the inorder successor till reaching the node $v$ with range perfectly contains $q$. Now the answer is the union of list of nodes traversed from $u$ to $v$. The time required for the same is $O(\log n +k)$, where $k$ is the number of nodes containing answer to the query.
\section{Comparison}
The static version of the standard segment tree of the segments, given in figure \ref{f1}, is shown in figure \ref{f4}. Similarly, the best dynamic version of the standard segment tree for the same set of segments is shown in figure \ref{f5}. For this particular example, the comparison is shown in table \ref{t1}. 
\begin{figure}[!ht]
\begin{center}
\tiny
\input{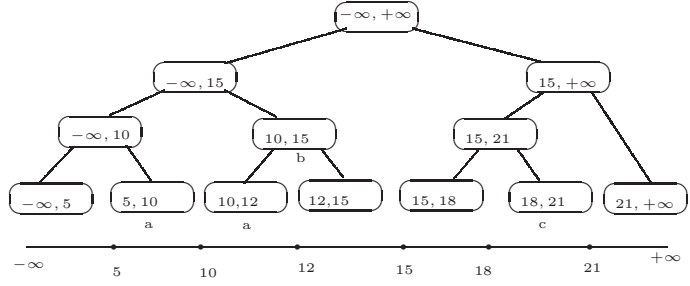}
\normalsize
\end{center}
\caption{Static version of the standard segment tree for the segments shown in figure \ref{f1}.}
\label{f4}
\end{figure}
\begin{figure}[!ht]
\begin{center}
\tiny
\input{fig2.pic}
\normalsize
\end{center}
\caption{Dynamic version of the standard segment tree for the segments shown in figure \ref{f1}.}
\label{f5}
\end{figure}
\begin{table}[!ht]
\caption{Comparison of trees given in figures \ref{f1}, \ref{f4} and \ref{f5}. SST denotes Static Segment Tree, DST denotes Dynamic Segment Tree and BITS denotes BITS Tree.}
\label{t1}
\vspace{3mm}
\begin{center}
\begin{tabular}{|l|c|c|c|}
\hline
Description & SST & DST & BITS\\
\hline
Number of nodes &13&31&4\\
\hline
Cumulative list size &4&8&5\\
(for segment storage)&&&\\
\hline
Height of the tree &3&4&2\\
\hline
Maximum stabbing&7&9&4\\
query time&&&\\
\hline
\end{tabular}
\end{center}
\end{table}
The maximum stabbing query time for the static segment tree comes when the query point is 15, and in the dynamic segment tree the maximum query time comes when the query point is 13. However, this analysis is not sufficient, and thus the factors related to the theoretical analysis  with respect to the number of segments as $n$ is shown in table \ref{t2}, which are due to the following lemmas.
\par
For the following lemmas, let us assume SST and DST respectively represent the static and dynamic versions of the standard segment trees. Also, the terms trees and BITS trees are used interchangeably in rest of the paper.
\begin{lemma}
\label{l1}
The maximum number of nodes required for the BITS tree is $2n-1$, where $n$ is the number of segments.
\end{lemma}
\begin{proof}
The range of nodes in BITS trees are due to the end points of the segments, and these ranges are not overlapped with the ranges of other nodes in the tree. As $n$ segments partition the range $[-\infty,+\infty]$ into $2n+1$ sub-ranges, the BITS tree can have the maximum of $2n-1$ nodes, since the BITS tree cannot have the nodes with range either $[-\infty,x]$ or $[y,+\infty]$, where $x$ and $y$ are respectively $\min _{1\le i\le n} \{x_i|s_i =[x_i,y_i]\}$ and $\max _{1\le i\le n} \{y_i|s_i =[x_i,y_i]\}$ for the set of segments $\{s_1,s_2, \cdots s_n\}$.
\end{proof}
\begin{table}[!ht]
\caption{Theoretical comparison of static segment trees, dynamic segment trees and the proposed BITS trees. SST denotes Static Segment Trees, DST denotes Dynamic Segment Trees and BITS denotes BITS Trees.}
\label{t2}
\vspace{3mm}
\footnotesize
\begin{center}
\begin{tabular}{|l|c|c|c|}
\hline
Description & SST & DST & BITS\\
\hline
Maximum&&&\\
number of nodes &4n+1&$\ge 2(n_2-n_1)-1$&2n-1\\
\hline
Maximum &&&\\
cumulative list &$2n\lceil\log (2n+1)\rceil-1$&$2n\lceil\log (n_2-n_1)\rceil-1$&$n^2$\\
size (for segment &&&\\
storage)&&&\\
\hline
Maximum height&&&\\
of the tree &$\lceil\log (2n+1)\rceil$&$\lceil\log (n_2-n_1)\rceil$&$1.441\lceil \log n\rceil$\\
\hline
Maximum stabbing&$2\lceil\log (2n+1)\rceil-1+k$&$2\lceil\log (n_2-n_1)\rceil-1+k$&$1.441\lceil \log n\rceil+1+k$\\
query time&&&\\
\hline
\end{tabular}
\end{center}
\normalsize
\end{table}
\begin{lemma}
Maximum storage required for maintaining the lists on all the nodes together in BITS tree is $n^2$.
\end{lemma}
\begin{proof}
The maximum storage required for maintaining the lists arises only when one of the node's (say, $N$) list with $n$ segments and two nodes, other than $N$, each with $n-1$, $n-2$, $\cdots$, 1 segments.
Thus, a total of $n+2((n-1)+(n-2)+\cdots+1)$, which is $n^2$, segments (may be with replication across the lists) stored altogether in the BITS tree. 
\end{proof}
\begin{lemma}
The height of the BITS trees does not exceed $\lceil \log n\rceil$.
\end{lemma}
\begin{proof}
The proof follows from the fact that the height of the AVL tree is less than ${1.441}\lceil \log n\rceil$ \cite{avl}.
\end{proof}
\begin{lemma}
Stabbing query on BITS tree can be addressed within $\lceil \log n\rceil+1+k$ time, where $k$ is the number of output segments.
\end{lemma}
\begin{proof}
As the height of the tree is $\lceil \log n\rceil$, to locate a node whose range contains the query point requires $O(\lceil \log n\rceil)$ time. Once such a node is found, the segments exists in the list of that node becomes the answer, and thus, we add $\theta(k)$ time to output the segments. However, if the query point is exactly one of end points of a segment, it is required to access either the inorder successor or the inorder predecessor of the node to list out the segments. In this case, it requires an additional $\theta(1)$ time to visit that node using the threads. Thus, the stabbing query can be addressed within the time of $\lceil \log n\rceil+1+k$.
\end{proof}
\begin{lemma}
Let $[n_1,n_2]$ be the range of root of a DST $T$, and $n$ be the number of segments stored in $T$. Then, $n \le n_2-n_1$.
\end{lemma}
\begin{proof}
Let $k$ be the number of leaf nodes in $T$. Let the range of $i$-th leaf be $[a_i,b_i]$, where  $1\le i\le k$. It is clear that the range $[a_i,b_i]$ of the leaf in DST satisfies $b_i-a_i=1$. Also, the number of leaf nodes in $T$ is $n_2-n_1$, which is the disjoint union of ranges of unit length across the range of the root. Thus, if $n > n_2-n_1$, then there should at least one segment, say $s_i=[x_1,x_2]$, exists by satisfying $x_2-x_1 <1$, which is impossible. Thus, $n \le n_2-n_1$.
\end{proof}
\par
\begin{lemma}
BITS trees outperforms SST and DST in case of range search.
\end{lemma}
\begin{proof}
Given a range, say $r=[x_1,x_2]$, the range search has to find the number of segments overlapped with $r$. In BITS tree, this is done by locating the node $N_1$, whose range contains $x_1$, then by visiting the nodes using inorder successor threads till reaching the node $N_2$ with the range containing $x_2$. Here, the segments stored in the list of nodes between $N_1$ and $N_2$ in the inorder sequence will be the answer. This process needs to visit $O(\log n +p)$ nodes, where $p$ is the number of nodes lies between $N_1$ and $N_2$ in the inorder sequence. Note that, $\log n+p \le m$, where $m$ is the number of nodes in the tree, as a node visited once cannot be visited again. 
\par
However, it is impossible to use inorder sequence in standard segment trees due to the following two reasons.
\begin{itemize}
	\item The SST or BST does not have inorder thread to efficiently access the next node in the inorder sequence.
	\item The range of the nodes are overlapped with its ancestors.
\end{itemize}
Thus, to do range search on SST or DST, one has to use a known binary tree traversal method to answer the query, and that takes the time of $O(m)$, where $m$ is the number of nodes in the tree. This is due to fact that the properties of segment trees are not useful, in this case, to find the range search. 
\par
Thus, the BITS trees outperforms both SST and DST in case of answering the range query.
\end{proof}
\par
The above analysis shows that BITS trees are better, as compared to the standard segment trees, except in list size requirements. Asymptotically, it is $O(n^2)$ for BITS trees and $O(n\log n)$ for the SST, and may be more for DST as the number of segments does not determine this. However, this is not a major issue, as in any search tree, the time taken to query processing is more important than the space requirements. Once the segments are stored, the overall complexity relies on the number of queries performed. Moreover, in BITS tree the answer (for stabbing query) lies in only one node, but in SST or DST it may spread in several nodes. Due to reduction of height and number of nodes in the tree, it is sure that the BITS trees are much better compared to the standard segment trees. 
\section{Conclusion}
The BITS-tree, proposed in this paper, is a balanced two-way inorder threaded segment tree, which is useful for storing geometric data objects available in one dimension. It handles both stabbing and range queries very efficiently than the standard segment tress. The BITS-trees are not restricted to any bound in range, unlike in the standard segment trees, and that is main advantage of this structure. Thus, one can insert a segment of any range without knowing the overall range the segment tree. Moreover, it is quite possible to generalize this to higher dimension with slight modifications in the data representation.
\bibliographystyle{plain}
\bibliography{hema}
\end{document}